\documentclass[16pt]{article}

\usepackage[english]{babel}

\usepackage[a4paper,top=2cm,bottom=2cm,left=3cm,right=3cm,marginparwidth=1.75cm]{geometry}

\usepackage{amsmath}
\usepackage{amsfonts}
\usepackage{amssymb}
\usepackage{amsthm}
\usepackage{mathalfa}
\usepackage{mathrsfs}
\usepackage{graphicx}
\usepackage{tikz}
\usepackage[colorlinks=true, allcolors=blue]{hyperref}
\usepackage{multirow}
\newtheorem{facts}{Fact}

\newtheorem{theorem}{Theorem}[section]
\newtheorem{corollary}{Corollary}[theorem]
\newtheorem{lemma}[theorem]{Lemma}
\newtheorem{definition}{Definition}
\usepackage[title]{appendix}
\usepackage{comment}
\usepackage{algorithm}
\usepackage{eqparbox}
\usepackage{algorithmicx}
\usepackage{algpseudocode}

\newcommand{\EP}[1]{\mathscr{E}_{#1}} 
\newcommand{\VP}[1]{\mathscr{V}_{#1}}
\usepackage[numbers]{natbib}
\usepackage{float}
\usepackage{hhline}

\newcommand*\samethanks[1][\value{footnote}]{\footnotemark[#1]} 
\algnewcommand{\algorithmicgoto}{\textbf{go to}}%
\algnewcommand{\Goto}[1]{\algorithmicgoto~\ref{#1}}%

\title{Finding Maximum Cliques in Large Networks}
\author{S.~Y. Chan\thanks{Deakin University, Geelong, Australia, School of Information Technology, Faculty of Science Engineering \& Built Environment, \underline{Australia}}        
    \and    
    K. Morgan\samethanks
    \and
    J. Ugon\samethanks[1]}

\date{}

\begin{document}
\maketitle

\begin{abstract}
There are many methods to find a maximum (or maximal) clique in large networks. Due to the nature of combinatorics, computation becomes exponentially expensive as the number of vertices in a graph increases. Thus, there is a need for  efficient algorithms to find a maximum clique. In this paper, we present a graph reduction method that significantly reduces the order of a graph, and so enables the identification of a maximum clique in graphs of large order, that would otherwise be computational infeasible to find the maximum. We find bounds of the maximum (or maximal) clique using this reduction. We demonstrate our method on real-life social networks and also on Erd\"{o}s-Renyi random graphs.
\end{abstract}

\section{Introduction}
Networks are everywhere, including biological, traffic, communication and social networks~\cite{costa2011,goldenberg2010,vandijk2012}. In recent years, the study of social networks have received increasing levels of attention~\cite{omalley2008}. Social network analysis uses graph theory to analyse social structure and relations among people and among groups in organisations~\cite{iniguez2020}. Important data or information can be obtained through social network analysis. Such data are often useful in fields such as marketing, economics and industrial engineering \cite{wasserman1989}.

Graphs are efficient tools for modelling relationships and the dynamics within such networks. In a graph, the vertices model entities and edges model relationships between entities. Many researchers are interested in studying substructures of networks, which is equivalent to analysing subgraphs in a graph. A \emph{subgraph} is a subset of vertices and edges, and an \emph{induced subgraph} is the subgraph formed by the subset of vertices and \emph{all} edges connecting pairs of vertices in that subset. 

Networks can be analysed by counting and classifying subgraphs \emph{within} a graph~\cite {aparicio2014parallel, bera_et_al,  chakaravarthy2016_counting, fomin2012, kloks2000, maugis2018, ribeiro2010parallel, ribeiro2019survey}. The subgraph counting problem~\cite{fomin2012, ribeiro2019survey} consists of counting the number of subgraphs that are isomorphic to a given pattern graph $H$ in a  graph $G$. Counting subgraphs is a hard problem, since it generalises the subgraph isomorphism problem, which is NP-complete~\cite{cook1971theorem}.  
Given the hardness of this problem, as many networks are increasingly large, there is a need for  efficient algorithms to count subgraphs in a timely manner, or even to obtain good bounds on counts of these subgraphs. A comprehensive comparison between exact and approximate subgraph counting algorithms is given in \cite{ribeiro2019survey}. 

In this paper, we introduce a novel graph reduction technique, which we use as a preliminary step to count complete subgraphs in large networks. In \cite{chan2022}, we generalised the concept of a \emph{rich-club} \cite{colizza2006,csigi2017,vaquero13,zhou2004}.
While the rich-club ranks vertices by their degree (i.e., the number of complete graphs of order 2 incident to the vertex), our generalisation of the rich-club ranks vertices by the number of complete graphs of order $r$ incident to the vertex.

We derived two new measures, namely, the \emph{vertex-participation of order $r$} and the \emph{edge-participation of order $r$} that count the number of complete subgraphs of order $r$ incident to a given vertex or edge respectively \cite{chan2022}. Our graph reduction algorithm uses the vertex- and edge-participation of small orders to identify vertices and edges that cannot belong to any clique of order $k$ or greater. These vertices and edges are removed as part of the graph reduction. We then recursively remove all vertices of degree less than $k-1$.

Thus, the graph $G$ is reduced to a graph of smaller order that still retains all cliques of order $\geq k$.  If the reduced graph is the empty graph, then no clique of size $k$ exists in $G$. Otherwise, the size of the maximum clique in the reduced graph provides a lower bound on the maximum clique in $G$. 

We apply this method to find large cliques in real-world social networks. We also present experimental results on the efficacy of our approach in finding the size of a maximum/maximal clique in randomly generated graphs. The reduction in the graph order enables the use of existing algorithms on graphs that would usually be infeasible.

This paper is organised as follows: Section \ref{section 2} gives some basic notations and definitions used in this paper. In Section \ref{section 3}, we describe our graph reduction. In Section \ref{section:SNAP}, we demonstrate our reduction method on real-life social network data. Section \ref{section 5} gives some results on Erd\"{o}s-R\'{e}nyi random graphs, and compare time taken to find a maximal clique in the original graphs and reduced graphs. Section \ref{conclusion} concludes the paper and discusses future work within this project. 

\section{Notation and Definitions}\label{section 2}

In this section, we provide necessary definitions and terminologies that are used throughout this paper. All graphs in this paper are \emph{simple} unless stated otherwise.

A graph $G$ is a pair $(V,E)$, such that $V$ is the (finite) set of vertices and $E$ is the set of edges. The \emph{order} of a graph refers to the number of vertices, whereas the \emph{size} of a graph refers to the number of edges. Let $u,v\in V(G)$, we say that $u$ is \emph{adjacent} to $v$ if there exists an edge $\{u,v\}\in E(G)$. We say that the edge $\{u, v\}$ is incident to vertices $u$ and $v$.

Let $G$ and $G_{1}$ be graphs of order $n$ and $k$ respectively, where
$n\geq k$. We say that $G_{1}$ is a \emph{subgraph} of $G$ if $V(G_{1})\subseteq V(G)$ and $E(G_{1})\subseteq E(G)$. The graph $G_{1}$ is an \emph{induced subgraph} of $G$ if all the edges between the pairs of vertices in $V(G_{1})$ from $E$ are in $E(G_{1})$, denoted $G_{1}\subseteq_{i} G$.

A complete graph of order $k$ is a graph where every pair of vertices are connected by an edge, denoted $K_{k}$. A \emph{clique} of size $k$ is a complete subgraph of order $k$ in a graph $G$.

The \emph{degree} of a vertex $v$ which we denote as $\delta(v)$ is the number of edges incident to $v$. The \emph{density} of a graph $G$ denoted $\rho(G)$ is the number of edges in $G$ over all possible edges, that is 
\begin{equation*}
\rho(G) = \dfrac{|E(G)|}{\binom{|V(G)|}{2}}=\dfrac{2|E(G)|}{|V(G)|^2-|V(G)|}
\end{equation*}
where $0\leq \rho(G) \leq 1.$

A universal vertex $v$ of a graph $G$ is a vertex that is connected to every other vertex in $G$. A universal vertex has degree $\delta(v)=|V(G)|-1$. A vertex $v$ \textit{belongs} to $K_{k}$ if $v\in V(K_{k})$.  Similarly, an edge $e$ \textit{belongs} to $K_{k}$ if $e\in E(K_{k}).$

\section{Graph Reduction} \label{section 3}

Due to the nature of combinatorics, counting cliques of order $k$ in a graph $G$ becomes computationally expensive as the order of $G$ and $k$ increase. Thus, it is infeasible to find a maximum clique in graphs of large order. If a large graph can be reduced to a graph of smaller order which retains the original maximum clique(s), it may be possible to reduce computation time and memory required to find a maximum clique. In this section, we introduce a graph reduction method that reduces the order of the graph $G$ without compromising the size of the maximum clique.

Several graph reduction techniques have been studied in order to reduce computation time. In most cases, a  graph of large order is reduced to a graph with far fewer vertices, while still retaining important information of the graph. One of the techniques for graph reduction is to use optimisation to reduce distance in graphs, or even incomplete \textbf{LU} factorisation~\cite{chen2022} (a technique in linear algebra used to decompose a matrix as the product of a lower triangular and upper triangular matrix).

Our reduction technique has three steps, which remove edges and vertices that do not belong in any  clique of order at least $k$. The first two steps use the concepts of \emph{vertex}- and \emph{edge-participation} of order $r$. These were first introduced in \cite{chan2022} in a generalisation of the rich-club.  The \emph{vertex-participation  of order $r$} of a vertex $v$ 
 is given by:
\[\VP{r}(v)=|\{G'\subseteq_{i}G:v\in V(G'),G'\cong K_{r}\}|.\]
Similarly, the \emph{edge-participation  of order $r$} of an edge $e$ 
is given by:
\[\EP{r}(e)=|\{G'\subseteq_{i}G:e\in E(G'),G'\cong K_{r}\}|.\]

We will show that any edge that has edge-participation of $\EP{r}(e)<\binom{k-2}{r-2}$ and any vertex of vertex-participation $\VP{r}(v)<\binom{k-1}{r-1}$ cannot belong in any clique of order $k \geq r$. Thus, any vertices and edges that do not meet these thresholds are removed from the graph in the first two steps of reduction.  

The final step in our reduction method is to recursively remove all vertices of degree $\delta(v)<k$. A \textit{$r$-core} of a graph $G$ is a maximal connected induced subgraph of $G$, such that all vertices have degree at least $r$. The third step essentially finds a $(k-1)$-core  of the graph obtained by the reduction in the first two steps.  We call the graph obtained by our reduction the \textit{$k$-nub}.  Figure \ref{fig:kcore} compares the $3$-core and the $4$-nub of a graph $G$. Although both these graphs have minimum degree $\geq 4$, the 3-core has over twice the vertices of the 4-nub.

\begin{definition}
The $k$-nub of a graph $G$ is the ($k-1$)-core of the reduced graph, obtained by deleting vertices with vertex-participation  $\VP{r}(v)<\binom{k-1}{r-1}$
and deleting edges with  edge-participation $\EP{r}(e)<\binom{k-2}{r-2}$.
\end{definition}

\begin{figure}[H]
    \centering
    \begin{tikzpicture}
     [scale=.7, auto=left, every node/.style={circle, fill=black}]
     \node[label=left:{0}] (n0) at (0,0) {};
     \node[label=left:{1}] (n1) at (0,-1.5) {};
     \node[label=left:{2}] (n2) at (0,-3) {};
     \node[label=above:{3}] (n3) at (1.5,0) {};
     \node[label=right:{4}] (n4) at (3,0) {};
     \node[label=95:{5}] (n5) at (1.5,-1.5) {};
     \node[label=45:{6}] (n6) at (3,-1.5) {};
     \node[label=below:{7}] (n7) at (1.5,-3) {};
     \node[label=below:{8}] (n8) at (3,-3) {};
     \node[label=left:{9}] (n9) at (3,1.5) {};
     \node[label=right:{10}] (n10) at (4.5,1) {};
     \node[label=30:{11}] (n11) at (4.5,-1.5) {};
     \node[label=270:{12}] (n12) at (4.5,-3) {};
     \node[label=-30:{13}] (n13) at (5.5,-4) {};
     \node[fill=none, draw=none] (n14) at (2.3,-5) {$G$};
     \draw (n0) edge (n1);
     \draw (n0) edge (n3);
     \draw (n0) edge (n5);
     \draw (n1) edge (n2);
     \draw (n1) edge (n5);
     \draw (n2) edge (n5);
     \draw (n2) edge (n7);
     \draw (n3) edge (n4);
     \draw (n3) edge (n5);
     \draw (n4) edge (n5);
     \draw (n4) edge (n6);
     \draw (n4) edge (n9);
     \draw (n4) edge (n10);
     \draw (n5) edge (n6);
     \draw (n5) edge (n7);
     \draw (n5) edge (n8);
     \draw (n6) edge (n8);
     \draw (n6) edge (n11);
     \draw (n6) edge (n12);
     \draw (n7) edge (n8);
     \draw (n8) edge (n11);
     \draw (n8) edge (n12);
     \draw (n8) edge (n13);
     \draw (n9) edge (n10);
     \draw (n11) edge (n12);
     \draw (n11) edge (n13);
     \draw (n12) edge (n13);
     
     \begin{scope}[yshift=-8cm, xshift=-5cm]
     \node[label=left:{0}] (n0) at (0,0) {};
     \node[label=left:{1}] (n1) at (0,-1.5) {};
     \node[label=left:{2}] (n2) at (0,-3) {};
     \node[label=above:{3}] (n3) at (1.5,0) {};
     \node[label=right:{4}] (n4) at (3,0) {};
     \node[label=95:{5}] (n5) at (1.5,-1.5) {};
     \node[label=45:{6}] (n6) at (3,-1.5) {};
     \node[label=below:{7}] (n7) at (1.5,-3) {};
     \node[label=below:{8}] (n8) at (3,-3) {};
     \node[label=30:{11}] (n11) at (4.5,-1.5) {};
     \node[label=270:{12}] (n12) at (4.5,-3) {};
     \node[label=-30:{13}] (n13) at (5.5,-4) {};
     \node[fill=none, draw=none] (n14) at (2.3,-5.5) {$3$-core of $G$};
     \draw (n0) edge (n1);
     \draw (n0) edge (n3);
     \draw (n0) edge (n5);
     \draw (n1) edge (n2);
     \draw (n1) edge (n5);
     \draw (n2) edge (n5);
     \draw (n2) edge (n7);
     \draw (n3) edge (n4);
     \draw (n3) edge (n5);
     \draw (n4) edge (n5);
     \draw (n4) edge (n6);
     \draw (n5) edge (n6);
     \draw (n5) edge (n7);
     \draw (n5) edge (n8);
     \draw (n6) edge (n8);
     \draw (n6) edge (n11);
     \draw (n6) edge (n12);
     \draw (n7) edge (n8);
     \draw (n8) edge (n11);
     \draw (n8) edge (n12);
     \draw (n8) edge (n13);
     \draw (n11) edge (n12);
     \draw (n11) edge (n13);
     \draw (n12) edge (n13);
     \end{scope}
     \draw[ultra thick,red, ->](0.5,-4.5) -- (-3,-7);
      \begin{scope}[xshift=5cm,yshift=-7cm]
     \node[label=45:{6}] (n6) at (3,-1.5) {};
     \node[label=below:{8}] (n8) at (3,-3) {};
     \node[label=30:{11}] (n11) at (4.5,-1.5) {};
     \node[label=270:{12}] (n12) at (4.5,-3) {};
     \node[label=-30:{13}] (n13) at (5.5,-4) {};
     \node[fill=none, draw=none] (n14) at (4,-6.5) {$4$-nub of $G$};
     \draw (n6) edge (n8);
     \draw (n6) edge (n11);
     \draw (n6) edge (n12);
     \draw (n8) edge (n11);
     \draw (n8) edge (n12);
     \draw (n8) edge (n13);
     \draw (n11) edge (n12);
     \draw (n11) edge (n13);
     \draw (n12) edge (n13);
     \end{scope}
     \draw[ultra thick,red, ->](4,-4.5) -- (8,-7.5);
    \end{tikzpicture}
    \caption{Example demonstrating the $(k-1)$-core and $k$-nub of $G$.}
    \label{fig:kcore}
\end{figure}
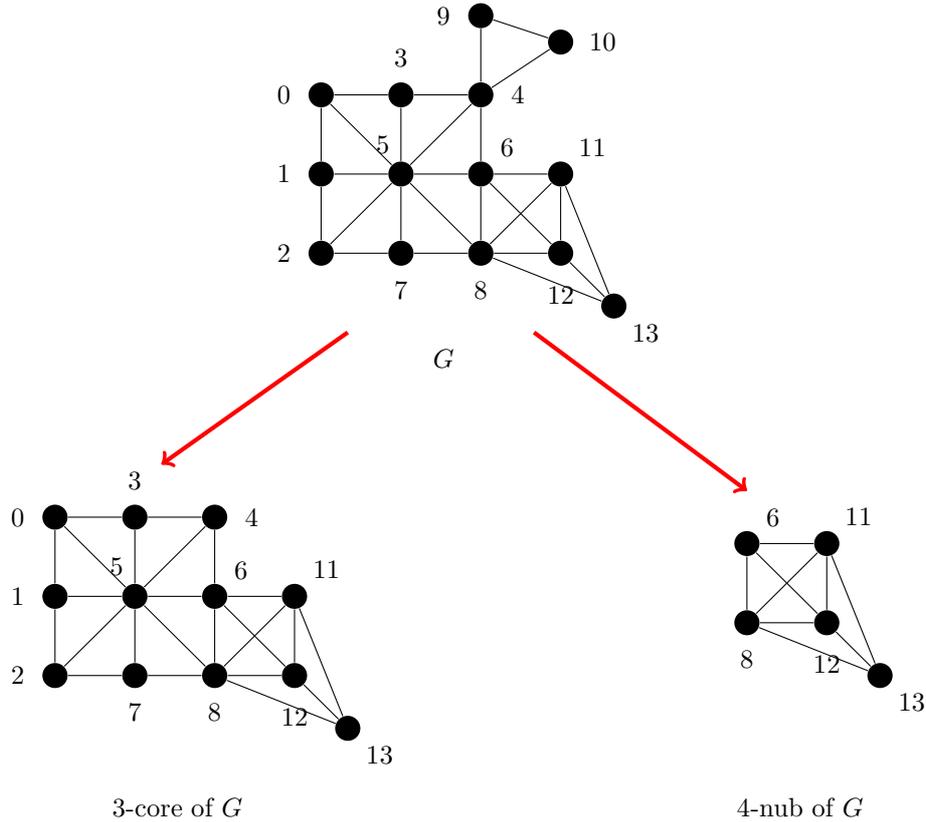

\subsection{The $k$-nub}
\noindent
We first introduce some facts and lemmas that are foundational for our graph reduction:

\begin{facts}
Any $k$-subset of vertices of a complete graph induces a complete subgraph. 
\label{fact1}
\end{facts}

\begin{facts}
If an edge does not belong to a complete subgraph of order $r<k-1$, then the edge will not belong to a complete graph of order $k$.
\end{facts}

\begin{facts}
If a vertex does not belong to a complete subgraph of order $r<k-1$, then the vertex will not belong to a complete graph of order $k$.
\end{facts}

\begin{facts}
If a vertex $v$ has degree $\delta(v)<k-1$, then it will not belong to any subgraph isomorphic to $K_{k}$.
\end{facts}

\begin{facts}
Any graph that has a clique of size $k$ has at least $\binom{k}{r}$ cliques of order $r\leq k$.
\label{fact5}
\end{facts}

\begin{lemma}
If an edge $e\in E(G)$ has $\EP{r}(e)<\binom{k-2}{r-2}$  then  $\EP{k}(e)=0$.
\label{EP}
\end{lemma}

\begin{proof}
Suppose an edge $e \in E(G)$ belongs to at least one clique of order $k$, then $\EP{r}(e)\geq \binom{k-2}{r-2}$. Since $\EP{r}(e)< \binom{k-2}{r-2}$, edge $e$ can not belong to any clique of order $k$ and so $\EP{r}(e)=0$.
\end{proof}

\begin{lemma}
If a vertex $v \in V(G)$ has $\VP{r}(v)<\binom{k-1}{r-1}$ then $\VP{k}(v)=0$.
\label{VP}
\end{lemma}

\begin{proof}
Suppose vertex $v\in V(G)$ belongs to at least one clique of order $k$, then 
$\VP{r}(v)\geq \binom{k-1}{r-1}$.  But $\VP{r}(v)<\binom{k-1}{r-1}$, and so $v$ cannot belong to any clique of order $k$ and  $\VP{k}(v)=0$.
\end{proof}

\begin{theorem}
If $G$ contains a clique of order at least $k'\geq k$, then the $k$-nub also contains a clique of order $k'$.
\end{theorem}

\begin{proof}
The reduction to obtain the $k$-nub has three steps: (1) removing edges with $\EP{r}(e)<\binom{k-2}{r-2}$,  
(2) removing vertices with $\VP{r}(v)<\binom{k-1}{r-1}$ and (3) finding a $(k-1)$-core of the graph obtained by the reduction in steps (1)-(2), for some $r\in [2,k-1]$.  We show that this reduction does not remove cliques of order $k'\geq k$ from the graph. 

If an edge $e$ belongs to a clique of order $k'\geq k$, then it has $\EP{r}(e)\geq \binom{k'-2}{r-2}\geq \binom{k-2}{r-2}$ and so $e$ will not be removed in the first step. 

If a vertex $v$ belongs to a clique of order $k'\geq k$, then it has $\VP{r}(v)\geq \binom{k'-1}{r-1}\geq \binom{k-1}{r-1}$ and so $v$ will not be removed in the second step.

Thus, the graph $G'$ obtained from the first two steps of the reduction must contain a clique of order $k'\geq k$ if and only if $G$ contains a clique of order $k'\geq k$. Any existing clique of order $k'$ in $G'$ has at least $k'$ vertices with degree at least $k'-1\geq k-1$. The final step in the reduction finds the $(k-1)$-core of $G'$, which contains a clique of order $k'\geq k$ if and only if $G'$ contains a clique of order $k'\geq k$. 
\end{proof}

\begin{corollary}
If the $k$-nub of $G$ is the empty graph, then $G$ contains no clique of order $\geq k$.
\end{corollary}

\begin{facts}
If there are $<k$ vertices of $\delta(v)\geq k-1$ in $G$, then $G$ has no clique of order $\geq k$.
\end{facts}

\subsection{Algorithm for finding a $k$-nub.}

In this section, we give an overview of the graph reduction algorithm including the choice of parameters used in the process. 

\subsubsection{Counting $r$-cliques}
Our graph reduction algorithm requires a \emph{pre-processing} step that involves counting $r$-cliques in the given graph $G$ for some small $r$. Using the counts of these $r$-cliques, we are able to calculate both the $\EP{r}(e)$ and $\VP{r}(v)$ in $G$. These counts are used in our graph reduction algorithm (see Algorithm \ref{algorithm1}).  

\subsubsection{An upper bound on the maximum clique}
The counts of small cliques in the graph can be used to obtain an initial value of $k$ for the maximum clique.  From Fact \ref{fact1}, if $G$ has a maximum clique of order $k$, then there are at least $\binom{k}{r}$ $r$-cliques for any $k>r$. We choose the largest $k$ that satisfies this inequality for our initial $k$.  However, it is possible to run this algorithm with an estimated $k$.

\subsubsection{The Algorithm }\label{sec:alg}

Our graph reduction algorithm (see Algorithm \ref{algorithm1}) takes as input the graph $G$, the vertex- and edge-participation of order $r$, the largest order $r$ of cliques counted in the pre-processing step and the value $k$ of the maximum clique. 

In the first step (Lines 7-11), we remove edges that do not meet the threshold of the edge-participation of order $r$. From Lemma \ref{EP}, these edges can not belong to any $k$-clique. 
 
Similarly, in the second step (Lines 12-17), we remove vertices that do not meet the threshold of the vertex-participation of order $r$. From Lemma \ref{VP}, these vertices can not belong to any $k$-clique.  

In this Step 3 (Lines 18-21), we find the $(k-1)$-core of the reduced graph.  In the final step of the reduction, we recursively remove vertices that have degree $< k-1$,  which results in the $k$-nub of our graph $G$. The reduction steps are highlighted in red in Algorithm \ref{algorithm1}.

\begin{algorithm}[H]
\caption{Graph reduction}
\label{algorithm1}
\begin{algorithmic}[1]
\Require Graph $G$, $\EP{r}$, $\VP{r}$, $k$, $r$
\Ensure Graph $G'=(V',E')$
\State{$n=|V|$}
\State{$E'=E$}
\State{$V'=V$}
\State{$G'=(V',E')$}
\State{e\_bound$=\binom{k-2}{r-2}$}
\State{v\_bound$=\binom{k-1}{r-1}$}
\For{$e$ \textbf{in} $E'$}
\If{$\EP{r}(e)<$e\_bound}
\State{Remove $e$ from $E'$}
\rlap{\smash{$\left.\begin{array}{@{}p{2em}}\\{}\\{}\\{}\end{array}\color{red}\right\}%
          \color{red}\begin{tabular}{l} Step 1.\end{tabular}$}}
\EndIf
\EndFor
\For{$v$ \textbf{in} $V'$}
\If{$\VP{r}(v)<$v\_bound}
\State{Remove all $e\in E'$ incident to $v$}
\rlap{\smash{$\left.\begin{array}{@{}p{1em}}\\{}\\{}\\{}\\{}\\{}\end{array}\color{red}\right\}%
          \color{red}\begin{tabular}{l} Step 2.\end{tabular}$}}
\State {Delete vertex $v$ from $V'$}
\EndIf
\EndFor
\While{$\exists v\in V'$ with $\delta(v)<k$}
\State{Remove all $e\in E'$ incident to $v$}\rlap{\smash{$\left.\begin{array}{@{}p{5em}}\\{}\\{}\\{}\end{array}\color{red}\right\}%
       \color{red}\begin{tabular}{l} Step 3.\end{tabular}$}}
\State{Delete vertex $v$ from $V'$}

\EndWhile
\State \Return{$G'$}
\end{algorithmic}
\end{algorithm}

Figure \ref{fig:reduction exp} illustrates the 3 main steps of the reduction. First, we count the number of $K_{3}$s in the graph $G$, followed by computing the $\VP{3}(v)$ and $\EP{3}(e)$ of $G$. In order to choose an initial $k$, we first select the largest $k$ based on the $r$-clique counts. In this case, we have $k=5$ as there are 17 ($>\binom{5}{3}=10$) $3$-cliques. We then reduce $k$ further using the edge-participation.  As no edges participate in at least $\binom{5-2}{3-2}=3$ cliques of order 3, we have $k<5$.  No further reduction is possible using the vertex-participation in this case. Thus, now we assume a maximum clique of order  $k=4$ in $G$.

In Step 1, we remove edges that do not participate in at least $\binom{4-2}{3-2}=\binom{2}{1}=2$ $K_{3}$s. In Step 2, we remove vertices that do not participate in at least $\binom{3}{2}=3$ $K_{3}$s. Lastly, we find the 3-core of the reduced graph obtained in Steps 1 and 2 to find the 4-nub of $G$. Finding a $3$-core without Steps 1 and 2 would result in a larger subgraph as it would only remove vertices 9 and 10.

\begin{figure}[H]
    \centering
    \begin{tikzpicture}
    [scale=.7,auto=left, every node/.style={circle, fill=black}]
     \node[label=left:{0}] (n0) at (0,0) {};
     \node[label=left:{1}] (n1) at (0,-1.5) {};
     \node[label=left:{2}] (n2) at (0,-3) {};
     \node[label=above:{3}] (n3) at (1.5,0) {};
     \node[label=right:{4}] (n4) at (3,0) {};
     \node[label=95:{5}] (n5) at (1.5,-1.5) {};
     \node[label=45:{6}] (n6) at (3,-1.5) {};
     \node[label=below:{7}] (n7) at (1.5,-3) {};
     \node[label=below:{8}] (n8) at (3,-3) {};
     \node[label=left:{9}] (n9) at (3,1.5) {};
     \node[label=right:{10}] (n10) at (4.5,1) {};
     \node[label=30:{11}] (n11) at (4.5,-1.5) {};
     \node[label=270:{12}] (n12) at (4.5,-3) {};
     \node[label=-30:{13}] (n13) at (5.5,-4) {};
     \node[fill=none, draw=none] (n14) at (2.3,-5.5) {$G$};
     \draw (n0) edge (n1);
     \draw (n0) edge (n3);
     \draw (n0) edge (n5);
     \draw (n1) edge (n2);
     \draw (n1) edge (n5);
     \draw (n2) edge (n5);
     \draw (n2) edge (n7);
     \draw (n3) edge (n4);
     \draw (n3) edge (n5);
     \draw (n4) edge (n5);
     \draw (n4) edge (n6);
     \draw (n4) edge (n9);
     \draw (n4) edge (n10);
     \draw (n5) edge (n6);
     \draw (n5) edge (n7);
     \draw (n5) edge (n8);
     \draw (n6) edge (n8);
     \draw (n6) edge (n11);
     \draw (n6) edge (n12);
     \draw (n7) edge (n8);
     \draw (n8) edge (n11);
     \draw (n8) edge (n12);
     \draw (n8) edge (n13);
     \draw (n9) edge (n10);
     \draw (n11) edge (n12);
     \draw (n11) edge (n13);
     \draw (n12) edge (n13);
     
     \begin{scope}[xshift=12cm]
     \node[label=left:{0}] (n0) at (0,0) {};
     \node[label=left:{1}] (n1) at (0,-1.5) {};
     \node[label=left:{2}] (n2) at (0,-3) {};
     \node[label=above:{3}] (n3) at (1.5,0) {};
     \node[label=right:{4}] (n4) at (3,0) {};
     \node[label=95:{5}] (n5) at (1.5,-1.5) {};
     \node[label=45:{6}] (n6) at (3,-1.5) {};
     \node[label=below:{7}] (n7) at (1.5,-3) {};
     \node[label=below:{8}] (n8) at (3,-3) {};
     \node[label=left:{9}] (n9) at (3,1.5) {};
     \node[label=right:{10}] (n10) at (4.5,1) {};
     \node[label=30:{11}] (n11) at (4.5,-1.5) {};
     \node[label=270:{12}] (n12) at (4.5,-3) {};
     \node[label=-30:{13}] (n13) at (5.5,-4) {};
     \node[fill=none, draw=none] (n14) at (2.3,-5.5) {Step 1: remove $\EP{3}(e)<2$};
     \draw (n0) edge (n5);
     \draw (n1) edge (n5);
     \draw (n2) edge (n5);
     \draw (n3) edge (n5);
     \draw (n4) edge (n5);
     \draw (n5) edge (n6);
     \draw (n5) edge (n7);
     \draw (n5) edge (n8);
     \draw (n6) edge (n8);
     \draw (n6) edge (n11);
     \draw (n6) edge (n12);
     \draw (n8) edge (n11);
     \draw (n8) edge (n12);
     \draw (n8) edge (n13);
     \draw (n11) edge (n12);
     \draw (n11) edge (n13);
     \draw (n12) edge (n13);
     \end{scope}
     \draw[ultra thick,red, ->] (7.5,-1.5) -- (9.5,-1.5);
     
     \begin{scope}[yshift=-9cm]
     \node[label=right:{4}] (n4) at (3,0) {};
     \node[label=95:{5}] (n5) at (1.5,-1.5) {};
     \node[label=45:{6}] (n6) at (3,-1.5) {};
     \node[label=below:{8}] (n8) at (3,-3) {};
     \node[label=30:{11}] (n11) at (4.5,-1.5) {};
     \node[label=270:{12}] (n12) at (4.5,-3) {};
     \node[label=-30:{13}] (n13) at (5.5,-4) {};
     \node[fill=none, draw=none] (n14) at (3.3,-6) {Step 2: remove $\VP{3}(v)<3$};
     \draw (n4) edge (n5);
     \draw (n5) edge (n6);
     \draw (n5) edge (n8);
     \draw (n6) edge (n8);
     \draw (n6) edge (n11);
     \draw (n6) edge (n12);
     \draw (n8) edge (n11);
     \draw (n8) edge (n12);
     \draw (n8) edge (n13);
     \draw (n11) edge (n12);
     \draw (n11) edge (n13);
     \draw (n12) edge (n13);
     \end{scope}
     \begin{scope}
         \draw[ultra thick,red, ->](10.5,-5.5) -- (7,-8.5);
     \end{scope}
     \begin{scope}[xshift=12cm,yshift=-9cm]
     \node[label=45:{6}] (n6) at (3,-1.5) {};
     \node[label=below:{8}] (n8) at (3,-3) {};
     \node[label=30:{11}] (n11) at (4.5,-1.5) {};
     \node[label=270:{12}] (n12) at (4.5,-3) {};
     \node[label=-30:{13}] (n13) at (5.5,-4) {};
     \node[fill=none, draw=none] (n14) at (4,-6) {Step 3: remove $\delta(v)<3$};
     \draw (n6) edge (n8);
     \draw (n6) edge (n11);
     \draw (n6) edge (n12);
     \draw (n8) edge (n11);
     \draw (n8) edge (n12);
     \draw (n8) edge (n13);
     \draw (n11) edge (n12);
     \draw (n11) edge (n13);
     \draw (n12) edge (n13);
     \end{scope}
     \draw[ultra thick,red, ->] (8.5,-11.5) -- (12,-11.5);
    \end{tikzpicture}
    \caption{Example demonstrating the 3 steps used to obtain a $k$-nub: (1) edge-deletion based on edge-participation, (2) vertex-deletion based on vertex-participation and (3) recursive vertex deletion based on vertex degree.}
    \label{fig:reduction exp}
\end{figure}
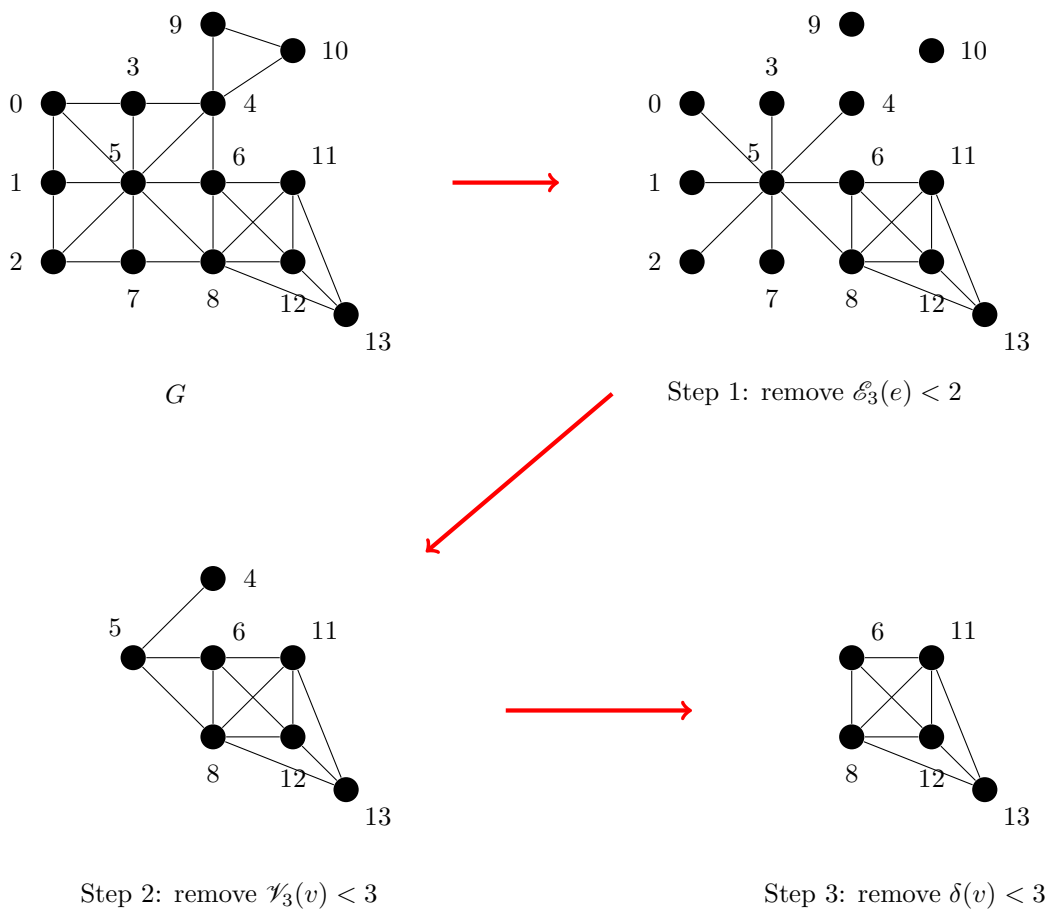

The algorithm has parameters   $r$ (the order of the vertex and edge participation) and $k$ (estimated maximum clique size).  We select the largest $r$ that is computational feasible to count $r$-cliques.  The selection of $k$ is more difficult. An upper bound on the size of $k$ can be determined using the counts of $r$-cliques, the graph must contain at least $\binom{k}{r}$  cliques of order $r$ if it has a $k$-clique. The largest $r$ value used in counting $r$-cliques provides a lower bound $l$ on the size of the maximum clique.

Looking at the reduced graph $G'$, we can refine these bounds.  A number of cases exist: 
\begin{enumerate}
    \item If $G'$ is an empty graph, then our original assumption that a maximum clique of order $k$ exists in $G$ is false. The estimated value of $k$ is too large in this case and we can reduce it.
    \item If the order of $G'$ is less than $k$, then our original assumption that a maximum clique of order $k$ exists in $G$ is false. The estimated value of $k$ can be reduced.  In addition, if the order of the maximal clique in $G'$ is $l'>l$ then we can increase the lower bound to $l'$.
    \item If the order of $G'$ is exactly $k$, then either $G'$ is the maximum clique of $G$ and we are done; or $G'$ has a maximal clique of size $l'$ and the size of the maximum clique in $G$ is in $[l', k-1].$
    \item If the order of $G'$ is greater than $k$, then the maximum clique in $G'$ is the maximum clique in $G$.  If $G'$ is too large to find a maximum clique, then the order $l'$ of a maximal clique in $G'$ provides a lower bound on the size of the maximum clique. 
\end{enumerate}
In practice, a reasonable guess for $k$ provides sufficient information to obtain an optimal graph reduction in two iterations using a \emph{binary search} approach to find a new value $k'$. We have

\begin{equation}
 k'=\left\lfloor\dfrac{l+k}{2}\right \rfloor.
\label{boundcount}
\end{equation}

\begin{lemma}
The \emph{time complexity} of Algorithm \ref{algorithm1} is $\mathcal{O}(n+m)$, where $n=|V(G)|$ and $m=|E(G)|$.
\end{lemma}
\begin{proof}
Suppose $k$ is the estimated order of the maximum clique. 
Lines 7-11 of Algorithm \ref{algorithm1} are executed at most $m$ times.  Each iteration includes 2 comparisons and a possible edge deletion.  Suppose $m'$ edges are deleted.   Then a total of at most $2m+m'$ steps are executed.   Lines 12-17 are executed at most $n$ times.  Each iteration includes 2 comparisons, and a possible vertex and at most $(k-1)$ edge deletions.   Suppose $n'$ vertices are deleted.  Then a total of at most $2n+n'(1+k-1)=2n+n'k$ steps are executed.  Finally, Lines 18-21 are executed at most $n-n'$ times.  Each iteration, includes 1 comparison and a possible vertex and at most $(k-1)$ edge deletions.  Suppose $n''$ vertices are deleted in this step.  Then a total of at most $(n-n')+(n-n')(k-1+1) = (n-n')(k+1)$  The total steps executed are $2m+m'+2n+n'k+nk-n'k+n-n'=2m+m'+3n+nk-n'< 3m+n(3+k).$
\end{proof}

\section{Results on Social Networks}
\label{section:SNAP}
One of the interests in social network analysis lies in identifying communities or groups formed within a network. The identification of communities in social networks have been a well studied problem~\cite{bedi2016}. This is known as the community detection problem. Communities can be linked to the classification of objects in categories for the sake of memorisation or retrieval of information. The saying of ``\emph{birds of the same feather flock together}", means that people with similar hobbies or tastes tend to form groups. Thus, we can use community detection to identify users with a high number of degrees (connections) and see how far their reach can travel in the network.

Communities in social networks can provide important information in many domains~\cite{bedi2016}. Through community detection, we can classify functions of `people' according to their structural positions in their identified communities. Using this information, we can also identify hierarchical organisations that exist in such networks. For example, communities in social networks are a representation of individuals with similar interests or tastes, and can reveal similar habits or patterns. This information can be particularly useful for people to form research collaborations, or even for marketing and purchase recommendations. 

One method of community detection is through the identification of cliques in the network~\cite{bedi2016}. Thus, finding the maximum clique in a graph is a useful application. In this section, we demonstrate the usefulness of our graph reduction method in finding the maximum clique in some social networks.

\subsection{Open Source SNAP and Network Repository}

We find the maximum clique in three social networks. The first social network dataset was obtained from \href{https://snap.stanford.edu/data/ego-Facebook.html}{SNAP} \cite{mcauley2012}. The authors of \cite{mcauley2012} developed a version of the Facebook (\emph{FB}) application and conducted a survey on users of the application. The data collection is based on \emph{categorising} friends into different \emph{social circles}. Examples of social circles are relatives, friends who share the same sports, or friends who attend the same university. We will refer to the FB graph as $G_{1}$.

We also ran our algorithm on two social networks obtained from another open source named  \href{https://networkrepository.com/}{Network Repository} \cite{nr}: (i) one on the mutually liked TV show pages (\emph{TV Shows}), and (ii) one on friendships and family relationships between users of a given website (\emph{Hamsterster}). This website provides a lower bound for the maximum clique in these graphs. Table \ref{tab:others summary} provides a summary of the information of these social networks. For simplicity, we will refer to the TV graph as $G_{2}$ and Hamsterster graph as $G_{3}$.

\begin{table}[H]
    \centering
    \begin{tabular}{|p{5cm}||r|r|r|}
    \hline
    \textbf{Statistics}  & \emph{FB} ($G_{1}$) & \emph{TV Shows} ($G_{2}$) & \emph{Hamsterster} ($G_{3}$) \\
    \hhline{|=||=|=|=|}
    Order of graph $|V(G)|$ & 4039 & 3900 & 2500 \\
    \hline 
     Size of graph $|E(G)|$ & 88234 &  17262 & 16630 \\ %
     \hline
    Maximum vertex degree & 1045 & 126 & 273 \\
    \hline
    Minimum vertex degree & 1 & 1 & 1 \\
    \hline
    Average vertex degree  & 43.69 & 4.43 & 6.65  \\
  (standard deviation) &  (52.42) & (12.55) & (19.72) \\
    \hline
    Diameter of graph & 8 & - & - \\
    \hline 
    Lower bound of maximum clique & - & 57 & 25 \\
    \hline
    \end{tabular}
    \caption{Summary of the social network data used in this paper, each labelled accordingly.}
    \label{tab:others summary}
\end{table}

\subsection{Method and Algorithm}
Algorithm \ref{algorithm1} requires parameters $k$ and $r$. The  initial value of $k$ was obtained by selecting the largest $k$ satisfying $\binom{k}{r}$ is at most the number of $r$-cliques in $G$. 
In this case, we use order $r$ for $r=5$. Table \ref{tab:time taken1} shows the exact counts and the time taken to compute for each $K_{r}$ in $G_{1}, G_{2}$ and $G_{3}$ respectively.  
\begin{table}[H]
    \centering
    \begin{tabular}{|c||c|r|r|r|}
    \hline 
    $r$ & \textbf{Graph} & \textbf{Time taken} (seconds) & \textbf{Number of $r$-cliques}  & Estimate $k$\\
    \hhline{|=||=|=|=|=|}
     \multirow{3}{*}{3}
     & $G_{1}$ & 47.90  & 1,612,010 & -\\
     & $G_{2}$ & 5.32 & 87,090 & -\\
     & $G_{3}$ & 5.97 & 53,251 & -\\
     \hline
     \multirow{3}{*}{4} 
     & $G_{1}$ &  852.46 & 30,004,668 & - \\
     & $G_{2}$ & 23.79 & 796,031 & - \\
     & $G_{3}$ & 19.09 & 132,809 & - \\
     \hline
     \multirow{3}{*}{5}
     & $G_{1}$ & 14724.89 & 517,944,123 & 109 \\
     & $G_{2}$ & 198.88 & 7,561,164 & 57\\
     & $G_{3}$ & 46.81 & 298,013 & 25\\
     \hline
    \end{tabular}
    \caption{Time taken to count the number of $K_{r}$s in $G_{1}\: (\text{FB}),\, G_{2}\: (\text{TV})$  and $G_{3}\: (\text{Hamsterster})$  including output file.}
    \label{tab:time taken1}
\end{table}

We first compute the vertex- and edge-participation of $r=5$ using the counts of $K_{5}$s in $G$. Algorithm \ref{algorithm1} finds the $k$-nub by removing  edges that do not meet the edge-participation threshold and vertices that do not meet the vertex-participation threshold, and then further reduces the graph by recursively removing vertices with degree $<k-1$. Our algorithm returns a much smaller graph of $n'<n$ vertices, which we then use to find the maximal or a maximum clique that appears in the original graph $G$ within a feasible amount of time.  

Our algorithm reduced the order of $G_{1}$ of order 4039 to 127 vertices ($\approx 97\%$), using $k=109$; the algorithm reduced the order of $G_{2}$ from 2500 to 25 vertices ($99\%$), using $k=25$, and it reduced the order of $G_{3}$ from 3900 to 61 vertices ($\approx 98\%$), using $k=57$. We were able to obtain the maximum clique for the latter two graphs and close bounds on the maximum clique of $G_{1}$ with $\omega \in [66,69].$

In some cases, the choice of $k$ for the reduction may  return the maximum clique in the graph. For $G_{1}$, the choice of $k$ differed to the order of the maximum clique that exists in $G_{1}$. However, for $G_{2}$ and $G_{3}$, when using a larger $k$ than the given bound, the reduction returned an empty graph. When using the given bound of $k=25$ for $G_{3}$, the graph reduced to a clique of order $25$, while $G_{2}$ reduced to an almost complete graph of order $61$ when using $k=57$. A summary of the results for all graphs are shown in Table \ref{tab:results tvhamster}. 

\begin{table}[H]
    \centering
    \begin{tabular}{|p{5cm}||r|r|r|}
    \hline
    \textbf{Results}  & \emph{FB} ($G_{1}$) & \emph{TV Shows} ($G_{2}$) & \emph{Hamsterster} ($G_{3}$) \\
    \hhline{|=||=|=|=|}
    Order of graph & 4039 & 3900 & 2500 \\
    \hline 
     Size of graph & 88234 &  17262 & 16630 \\ %
     \hline
     Choice of $k$ & 109 & 57 & 25 \\
    \hline
    Order of reduced graph & 127 & 61 & 25 \\
    \hline
    Size of reduced graph & 7634 & 1820  & 300 \\
    \hline
    Density of reduced graph & 0.9850 & 0.9945 & 1 \\
    \hline
    Maximum Clique & $66\leq \omega \leq 69$ & 57 & 25 \\
    \hline
    \end{tabular}
    \caption{Summary of results for graphs $G_{1}$, $G_{2}$ and $G_{3}$.}
    \label{tab:results tvhamster}
\end{table}

Finding the maximum clique is a NP-hard problem \cite{karp1972}. We ran the maximum clique algorithm  \cite{boppana1992,bron1973,tomita2006, cazals2008} as implemented in \emph{NetworkX} \cite{SciPyProceedings_11}  on all three social network graphs. However, when running on $G_{1}$, the program ran for several days and then crashed due to insufficient memory in the computer. A similar problem occurred, when using the \emph{maximal} clique algorithm \cite{boppana1992,bron1973} implemented in NetworkX.

\section{Results on Erd\"{o}s-R\'{e}nyi Random Graphs}\label{section 5}

We applied our graph reduction on random graphs of different orders and densities. We show the time taken to reduce $G$ using our method, and also compare the computation time for $r$-cliques in each graph as $r$ increases for $r=3,4,5$. We generated $10$ random graphs for each density $\rho(G)=0.1,0.3,0.4$ with order $n=|V(G)|=1000,2000,4000$. 

During the pre-processing stage, as $\rho(G)$ increases with $N$, it became computationally more expensive (i.e., computer memory and time taken) even when counting $K_{4}$s in graphs of $n=1000$ and $\rho(G)=0.1$. Thus, we only count cliques of order $r=3$ and record the time taken, and also the time taken to reduce all graphs in $G$ for all $n$. We further compared the time taken to find the maximum/maximal clique in both $G$ and reduced graph $G'$. The table below shows the average time taken to count $K_{3}$s for each $G$. 

\begin{table}[H]
    \centering
    \begin{tabular}{|c|c|r|c|}
    \hline
     $|V(G)|$  & $\rho(G)$ & Average Time (seconds) & $N$  \\
     \hhline{|=|=|=|=|}
      \multirow{3}{*}{1000} & 0.1 & 10 &  10\\ 
    & 0.3 & 20 & 10 \\ 
    & 0.4  & 60& 10\\ 
    \hline
     \multirow{3}{*}{2000} & 0.1 & 60 &10\\ 
    & 0.3 & 180 &10 \\ 
    & 0.4 & 300 &10\\
    \hline
     \multirow{3}{*}{4000} & 0.1 & 300 &10 \\
    & 0.3 & 1380 & 10 \\
    & 0.4 & 6000 & 10\\
    \hline
    \end{tabular}
    \caption{Time taken on average to count $K_{3}$s in the Erd\"{o}s-R\'{e}nyi random graphs.}
    \label{K3 time}
\end{table}

We ran our graph reduction algorithm on each of these randomly generated graphs using their respective $K_{3}$ counts, and also the vertex- and edge-participation of order 3. The approximate time taken for the reduction for each graph is given in the Appendix section. Table \ref{summary of results for random graph} summarises these results.

\begin{table}[H]
    \centering
    \begin{tabular}{|c|c|p{0.1\linewidth}|p{0.13\linewidth}|r|p{0.14\linewidth}|p{0.06\linewidth}|p{0.15\linewidth}|}
    \hline
     \multirow{2}{*}{$|V(G)|$} & \multirow{2}{*}{$\rho(G)$} & Reduction &   Reduced \%  & \multirow{2}{*}{$k$} 
     & NetworkX $G$ & Max &  NetworkX $G'$ \\
      &  & $\mu_{t}$, ($\sigma$) & $\mu_{\%}$, ($\sigma$) &  & $\mu_{t}$, ($\sigma$) & Clique & $\mu_{t}$, ($\sigma$) \\
     \hhline{|=|=|=|=|=|=|=|=|}
      \multirow{6}{*}{1000} & \multicolumn{1}{r|}{0.1} & \multicolumn{1}{r|}{11.14s} & \multicolumn{1}{r|}{0.385} & \multicolumn{1}{r|}{15} & \multicolumn{1}{r|}{119.34s} & \multicolumn{1}{r|}{5} & \multicolumn{1}{r|}{22.12s} \\ 
      & & \multicolumn{1}{r|}{(0.43s)} & \multicolumn{1}{r|}{(0.183)} & & \multicolumn{1}{r|}{(1.94s)} & & \multicolumn{1}{r|}{(0.98s)} \\
      & \multicolumn{1}{r|}{0.3} & \multicolumn{1}{r|}{11.08s} & \multicolumn{1}{r|}{0.338}  & \multicolumn{1}{r|}{93}  & \multicolumn{1}{r|}{216.95s} & \multicolumn{1}{r|}{9} & \multicolumn{1}{r|}{27.62s}  \\
      & & \multicolumn{1}{r|}{(0.25s)} & \multicolumn{1}{r|}{(0.050)} & & \multicolumn{1}{r|}{(1.69s)} & & \multicolumn{1}{r|}{(1.58s)} \\
      & \multicolumn{1}{r|}{0.4} & \multicolumn{1}{r|}{11.40s} & \multicolumn{1}{r|}{0.327} & \multicolumn{1}{r|}{161} & \multicolumn{1}{r|}{2802.65s} & \multicolumn{1}{r|}{11} & \multicolumn{1}{r|}{99.49s} \\
      & & \multicolumn{1}{r|}{(0.59s)} & \multicolumn{1}{r|}{(0.042)} & & \multicolumn{1}{r|}{(27.66s)} & & \multicolumn{1}{r|}{(1.95s)} \\
      \hline
      \multirow{6}{*}{2000} & \multicolumn{1}{r|}{0.1} & \multicolumn{1}{r|}{20.07s} & \multicolumn{1}{r|}{0.318}  &\multicolumn{1}{r|}{26} & \multicolumn{1}{r|}{845.38s}  & \multicolumn{1}{r|}{5} & \multicolumn{1}{r|}{376.48s} \\ 
      & & \multicolumn{1}{r|}{(0.49s)} & \multicolumn{1}{r|}{(0.039)} & & \multicolumn{1}{r|}{(89.52s)} & & \multicolumn{1}{r|}{(7.50s)} \\
      & \multicolumn{1}{r|}{0.3} & \multicolumn{1}{r|}{20.04s} & \multicolumn{1}{r|}{0.272} & \multicolumn{1}{r|}{183} & \multicolumn{1}{r|}{5183.24s} & \multicolumn{1}{r|}{10} & \multicolumn{1}{r|}{372.54s} \\
      & & \multicolumn{1}{r|}{(0.54s)} & \multicolumn{1}{r|}{(0.038)} & & \multicolumn{1}{r|}{(17.98s)} & & \multicolumn{1}{r|}{(32.41s)} \\
      & \multicolumn{1}{r|}{0.4} & \multicolumn{1}{r|}{20.07s} & \multicolumn{1}{r|}{0.341} & \multicolumn{1}{r|}{320} & \multicolumn{1}{r|}{NA} & \multicolumn{1}{r|}{12} & \multicolumn{1}{r|}{2761.16s} \\
      & & \multicolumn{1}{r|}{(0.61s)} & \multicolumn{1}{r|}{(0.039)} & &  & & \multicolumn{1}{r|}{(50.57s)} \\
      \hline
      \multirow{6}{*}{4000} & \multicolumn{1}{r|}{0.1} & \multicolumn{1}{r|}{80.27s} & \multicolumn{1}{r|}{0.407}  & \multicolumn{1}{r|}{48} & \multicolumn{1}{r|}{6416.01s}  & \multicolumn{1}{r|}{5} & \multicolumn{1}{r|}{2237.91s}  \\ 
      & & \multicolumn{1}{r|}{(1.14s)} & \multicolumn{1}{r|}{(0.180)} & & \multicolumn{1}{r|}{(139.49s)} & & \multicolumn{1}{r|}{(67.59s)} \\
      & \multicolumn{1}{r|}{0.3} & \multicolumn{1}{r|}{86.92s} & \multicolumn{1}{r|}{0.375} & \multicolumn{1}{r|}{365} & \multicolumn{1}{r|}{NA} & \multicolumn{1}{r|}{10} & \multicolumn{1}{r|}{2476.00s} \\
      & & \multicolumn{1}{r|}{(2.11s)} & \multicolumn{1}{r|}{(0.046)} & & & & \multicolumn{1}{r|}{(195.85s)} \\
      & \multicolumn{1}{r|}{0.4} & \multicolumn{1}{r|}{82.38s} & \multicolumn{1}{r|}{0.321} & \multicolumn{1}{r|}{638} & \multicolumn{1}{r|}{NA} & \multicolumn{1}{r|}{12} & \multicolumn{1}{r|}{NA} \\
      & & \multicolumn{1}{r|}{(0.71s)} & \multicolumn{1}{r|}{(0.033)} & & & &  \\
      \hline
    \end{tabular}
    \caption{Results obtained (on average, $\mu$) of the time taken $t$, and the percentage ($\%$) of the graph reduced by (on average, $\mu$) from the Erd\"{o}s-R\'{e}nyi random graphs $G$ and reduced graphs $G'$. }
    \label{summary of results for random graph}
\end{table}

We compare the time $t$ taken to find the maximum clique, or in some cases a maximal clique, in both $G$ and the reduced graph $G'$. We use the existing NetworkX implementation to find the maximum clique of $G$ and $G'$. Further details of the recorded time can also be found in Appendix. 

In cases where NetworkX was unable to find a maximum clique (i.e., inefficient computation time), we were able to use a brute force method to find a maximal clique in the reduced graph $G'$ for $n\!>2000$ that can be found within seconds. Our brute force method is a greedy algorithm (Algorithm \ref{algorithm2}) which takes in an input $r$ and checks all possible neighbouring vertices. We have demonstrated that the reduction algorithm enables the identification of a maximal clique that exists in a large graph within a feasible computation time.

\begin{algorithm}[H]
\caption{Brute Force Algorithm}
\label{algorithm2}
\begin{algorithmic}[1]
\Require Graph $G=(V,E)$, $r$, $\delta(v)$
\Ensure Maximal Clique
\State{bestClique = $r$}
\State{$S=[\:]$}
\State{best = $len(\text{bestClique})$}
\For{$v$ \textbf{in} $V$}
\If{$\delta(v)>$best}
\State{$S=[v]$}
\For{$w$ \textbf{in} $V\setminus S$}
\If{$N(w) \subseteq S$}
\State{S.append($w$)}
\If{$len(S)>$best}
\State{best=$len(S)$}
\State{bestClique=$S$}
\EndIf
\EndIf
\EndFor
\EndIf
\EndFor
\State \Return{best}
\end{algorithmic}
\end{algorithm}     

Table \ref{summary of results for random graph} shows that there is a significant reduction in the time taken to find a maximum clique in the reduced graph $G'$ than in $G$. In most cases, the time taken $t$ to find a maximum clique in $G$ and $G'$ were reduced by between $\approx 65\%$ and $\approx 88\%$. In cases where the order of the graph is larger and denser (i.e., $|V(G)|=2000,4000,\rho(G)=0.2,0,4$), NetworkX  was unable to find the maximum clique of $G$ within a feasible time. However, NetworkX was able to find a maximum clique in the reduced graph $G'$. In the only case where NetworkX was unable to find the maximum clique in $G'$ (i.e., $|V(G)|=4000$, $\rho(G)=0.4$), we find a maximal clique instead. Results are shown in the Appendix.

Our algorithm has also demonstrated a significant reduction in the order of the graphs. Table \ref{summary of results for random graph} shows that the $k$-nub contains between $\approx 60\%$ to $\approx 73\%$ of the vertices of  the original  graphs $G$ with varying orders. Upon reduction for graphs of order $<1000$, NetworkX was able to quickly compute the maximum clique with an average computation time of 22s. There was also a significant reduction in the time taken to find a maximum clique for the reduced graphs of higher order.

\subsection{The $k$-nub and the $(k-1)$-core}

The $k$-nub can significantly reduce the graph compared to simply finding a $(k-1)$-core itself.  We compared the $(k-1)$-core to the $k$-nub of $G$. We use the same $k$ value for finding the $k$-nub to show the significance in the graph reduction. We found that by using our known $k$ value to find the $(k-1)$-core for each of the random graph, this resulted in the original graph, which states that our $k$ value is too small for the $(k-1)$-core. 

Note that in order to find the $(k-1)$-core, this requires prior knowledge of a $k$ value. However, NetworkX also has an implemented $K$-core algorithm which returns the maximal connected induced subgraph of $G$, such that all vertices are degree at least $K$. Note that the $K$ value used in NetworkX is not the same as the $k$ value that we have found. The $K$-core found by using the NetworkX implementation only removed between 4 and 86 vertices ($<5\%$) amongst all random graphs.  

 \begin{table}[H]
 \centering
 \begin{tabular}{|c||c|c|c|}
 \hline
   & $\rho(G)=0.1$ & $\rho(G)=0.3$ & $\rho(G)=0.4$\\
  \hhline{|=||=|=|=|}
  $n=1000$ & 63.53\% & 67.17\% & 68.35\% \\
  ($\sigma$) & (18.66\%) & (5.28\%) & (4.24\%)\\
  \hline
  $n=2000$ & 70.11\% & 73.66\% & 66.34\% \\ 
  ($\sigma$) & (4.03\%) & (3.95\%) & (3.92\%) \\
  \hline
 $n=4000$ & 66.23\% & 63\% & 68.37\% \\
 ($\sigma$) & (3.68\%) & (4.65\%) & (3.36\%) \\
\hline
 \end{tabular}
 \caption{The percentage \% of vertices (on average) of the $k$-nub in the $k$-core.}
 \label{kcoreresults}
\end{table}

Table \ref{kcoreresults} shows on average the percentage of vertices of the $k$-nub in the $k$-core. This shows that the $k$-nub was able to reduce the graph by a further 33\% on average, compared to the reduction based on finding the $k$-core alone. In random graphs where the vertex degrees are uniformly distributed, the $k$-core itself was not sufficient to reduce the graph to a more manageable size. 

\section{Conclusion and Future Work}\label{conclusion}

In this paper, we have introduced the concept of a $k$-nub, a reduction of the original graph based on \emph{vertex}- and \emph{edge-participation} of order $r$ in a graph $G$.  We demonstrate that the $k$-nub can be used to find a maximum or a maximal clique in the original graph.   The significant reduction in order of the original graph, enabled these cliques to be identified in cases where the graph is too large for existing algorithms, such as NetworkX, to run in a feasible amount of time and space.

Results on three real-world social network graphs  showed the $k$-nub contained at most $3\%$ of the  vertices in the original graph and, thus, there was a significant reduction in the time taken to find a maximum or a maximal clique in $G$.  
Experiments on random graphs showed that the $k$-nub contained approximately $41\%$ to $88\%$ of the vertices of the original graph. 

In comparison to the existing \emph{NetworkX} implementation of the \emph{maximum clique} algorithm to find a maximum clique which took hours, or  returned `Memory Error', to compute for $n>\!2000$ (depending on density of the graph), our method showed significant improvement in computation time. 

Our experiments in random networks demonstrate that the $k$-nub is significantly smaller than the $(k-1)$-core alone. Using our $k$ value to find the $(k-1)$-core alone, this resulted in the original graph. Thus, we use the NetworkX implemented $K$-core algorithm to find the $K$-core of each of the random graphs and compared this to the $k$-nub. Note that $K$ and $k$ are different in this case. The $K$-core of all graphs only reduced the graphs by at most 6\%, while the $k$-nub gave a reduction of between 11\% and 59\%. 

Our experiments suggest that using the $k$-nub to identify large cliques is particularly useful in large graphs with low density that contain many  \emph{communities} (or clusters), as commonly seen in social networks. This observation is based on the percentage of reduction: in social networks, there is about a 98\% reduction compared to the random graphs, which were about 35\% reduction on average for all $N$.

In future work, we will look at first finding a $k$-core $H$ for some $k$ and then obtaining the vertex- and edge-participation in $H$ and then find a $k$-nub of $H$. This may reduce the time required to count the small $r$-cliques. 

Further, we may also consider using a similar reduction method on \emph{non-complete} graphs or graphs that are near complete (complete graphs with a few edges removed) to obtain (bounds) on the maximum/maximal clique. 

\begin{appendices} 

\section{Details of Results}

\begin{table}[H]
    \centering
    \begin{tabular}{|c|c|p{0.1\linewidth}|c|c|p{0.1\linewidth}|p{0.1\linewidth}|p{0.1\linewidth}|p{0.1\linewidth}|}
    \hline
     $|V(G)|$  & $\rho(G)$ & Reduction Time ($G$) &  $|V(G')|$  & $k$ & NetworkX Time ($G$) & Maximum Clique & NetworkX Time ($G'$) & $|K-\text{core}|$ \\
     \hhline{|=|=|=|=|=|=|=|=|=|}
    \multirow{10}{*}{1000} & 0.1 & 10.52s & 483 & 15 & 118.58s  & 6 & 22.34s & 938 \\ 
    & 0.1 & 11.06s & 886 & 14 & 119.56s & 6 &  23.06s & 979\\ 
    & 0.1 & 11.20s & 880 & 14 & 117.34s & 6 &  21.63s & 981\\ 
    & 0.1 & 10.81s & 604 & 15 & 120.56s & 5 &  20.56s & 959\\ 
    & 0.1 & 11.32s & 493 & 15 & 123.55s & 6 &  24.12s & 967\\ 
    & 0.1 & 11.76s & 872 & 14 & 118.33s & 6 &  22.43s & 960\\ 
    & 0.1 & 10.75s & 411 & 15 & 121.56s & 6 &  22.36s & 972\\ 
    & 0.1 & 11.58s & 580 & 15 & 118.65s & 5 &  20.78s & 972\\ 
    & 0.1 & 10.69s & 531 & 15 & 118.39s & 5 &  22.04s & 957\\ 
    & 0.1 & 11.74s & 410 & 15 & 116.83s & 5 &  21.87s & 990\\ 
    \hline
    \multirow{10}{*}{1000} & 0.3 & 215.65s & 649 & 93 & 28.15s & 9 & 11.56s & 988\\
    & 0.3 & 10.69s &  689 & 93 & 217.05s & 8 & 25.54s & 984\\
    & 0.3 & 11.24s &  662 & 93 & 217.80s & 9 & 26.33s & 985\\
    & 0.3 & 10.57s &  565 & 93 & 214.93s & 8 & 25.12s & 988\\
    & 0.3 & 11.26s &  684 & 93 & 217.35s & 9 & 28.67s & 974\\
    & 0.3 & 10.96s &  683 & 93 & 217.76s & 9 & 26.73s & 993\\
    & 0.3 & 11.15s &  595 & 93 & 214.15s & 9 & 28.47s & 992\\
    & 0.3 & 10.94s &  748 & 93 & 220.46s & 9 & 30.56s & 977\\
    & 0.3 & 11.05s &  699 & 92 & 217.86s & 8 & 28.68s & 987\\
    & 0.3 & 11.13s &  645 & 93 & 216.44s & 9 & 27.92s & 990\\
    \hline
    \multirow{10}{*}{1000} & 0.4 & 11.03s & 629 & 161 & 2776.76s & 11 & 97.56s & 994\\
    & 0.4 & 12.45s & 627 & 161 & 2794.05 & 11 & 100.06ss & 989\\
    & 0.4 & 11.29s & 669 & 160 & 2847.62s & 11 & 100.33s & 986\\
    & 0.4 & 11.59s & 653 & 161 & 2784.72s & 11 & 96.94s & 982\\
    & 0.4 & 10.67s & 721 & 160 & 2768.33s & 11 & 99.66s & 980\\
    & 0.4 & 12.21s & 720 & 159 & 2831.53s & 11 & 100.92s & 989\\
    & 0.4 & 10.84s & 602 & 161 & 2775.18s & 11 & 96.23s & 970\\
    & 0.4 & 11.75s & 724 & 160 & 2844.09s & 11 & 102.58s & 986\\
    & 0.4 & 10.64s & 674 & 160 & 2795.93s & 11 & 99.02s & 982\\
    & 0.4 & 11.54s & 709 & 160 & 2808.30s & 11 & 101.58s & 985\\
    \hline
    \end{tabular}
    \caption{The table shows time taken using our graph reduction algorithm for $n=1000$ and $\rho(G)=0.1,0.3,0.4$ respectively. Time taken for NetworkX to find the maximum clique is also shown.}
    \label{tab:results for 1000}
\end{table}

\begin{table}[H]
    \centering
    \begin{tabular}{|c|c|p{0.1\linewidth}|c|c|p{0.1\linewidth}|p{0.1\linewidth}|p{0.1\linewidth}|p{0.1\linewidth}|}
    \hline
     $|V(G)|$  & $\rho(G)$ & Reduction Time ($G$) &  $|V(G')|$  & $k$ & NetworkX Time ($G$) & Maximum Clique & NetworkX Time ($G'$) & $|K-\text{core}|$  \\
     \hhline{|=|=|=|=|=|=|=|=|=|}
     \multirow{10}{*}{2000} & 0.1 & 19.64s & 1391 & 26 & 874.59s & 5 & 388.56s & 1946\\
    & 0.1 & 20.22s & 1447 & 26 & 890.35s & 5 &  368.86s & 1944\\
    & 0.1 & 19.10s & 1455 & 26 & 901.24s & 6 & 386.41s & 1962\\
    & 0.1 & 20.56s & 1362 & 26 & 886.82s & 5 & 377.16s & 1953\\
    & 0.1 & 20.78s & 1360 & 26 & 870.87s & 5 & 376.03s & 1934\\
    & 0.1 & 20.11s & 1174 & 26 & 850.79s & 6 & 360.57s & 1957\\
    & 0.1 & 19.75s & 1433 & 26 & 846.71s & 5 & 378.40s & 1950\\
    & 0.1 & 19.89s & 1358 & 26 & 866.25s & 5 & 376.33s & 1915\\
    & 0.1 & 20.64s & 1309 & 26 & 860.18s & 5 & 375.15s & 1959\\
    & 0.1 & 19.97s & 1359 & 26 & 873.02s & 5 & 377.28s & 1948\\
    \hline
     \multirow{10}{*}{2000} & 0.3 & 20.12s & 1355 & 183 & 5153.89s & 9 & 373.30s & 1975 \\
    & 0.3 & 19.57s & 1524 & 183 & 5200.69s & 10 & 395.08s & 1962\\
    & 0.3 & 19.25s & 1407 & 183 & 5191.42s & 10 & 278.25s & 1986\\
    & 0.3 & 19.79s & 1403 & 183 & 5174.67s & 10 & 375.85s & 1988\\
    & 0.3 & 20.66s & 1546 & 183 & 5197.96s & 9 & 384.49s & 1976\\
    & 0.3 & 20.47s & 1341 & 183 & 5151.10s & 10 & 370.16s & 1979\\
    & 0.3 & 20.83s & 1520 & 183 & 5203.52s & 9 & 388.79s & 1978\\
    & 0.3 & 19.22s & 1518 & 183 & 5192.65s & 10 & 389.82s & 1973\\
    & 0.3 & 20.34s & 1409 & 183 & 5173.78s & 10 & 378.50s & 1972\\
    & 0.3 & 20.11s & 1532 & 183 & 5192.65s & 9 & 391.17s & 1975\\
    \hline
    \multirow{10}{*}{2000} & 0.4 & 20.95s & 1177 & 320 & Memory Error & 12 & 2659.71s & 1987\\
    & 0.4 & 20.38s & 1345 & 320 & Memory Error & 12 & 2789.38s & 1984\\
    & 0.4 & 19.91s & 1388 & 320 & Memory Error & 12 & 2798.21s & 1995\\
    & 0.4 & 19.89s & 1349 & 320 & Memory Error & 12 & 2756.96s & 1982\\
    & 0.4 & 19.24s & 1349 & 320 & Memory Error & 12 & 2774.23s & 1987\\
    & 0.4 & 19.11s & 1441 & 320 & Memory Error & 12 & 2840.09s & 1990\\
    & 0.4 & 21.03s & 1257 & 320 & Memory Error & 12 & 2712.30s & 1996\\
    & 0.4 & 19.77s & 1382 & 319 & Memory Error & 12 & 2812.19s & 1980\\
    & 0.4 & 20.44s & 1236 & 320 & Memory Error & 12 & 2722.89s & 1972\\
    & 0.4 & 20.02s & 1250 & 319 & Memory Error & 12 & 2745.67s & 1985\\ 
    \hline
    \end{tabular}
    \caption{Table shows time taken using our graph reduction algorithm for $n=2000$ and $\rho(G)=0.1,0.3,0.4$ respectively. Time taken for NetworkX to find the maximum clique is also shown.}
    \label{tab:results for 2000}
\end{table}

\begin{table}[H]
    \centering
    \begin{tabular}{|c|c|p{0.1\linewidth}|c|c|p{0.1\linewidth}|p{0.1\linewidth}|p{0.1\linewidth}|p{0.1\linewidth}|}
    \hline
     $|V(G)|$  & $\rho(G)$ & Reduction Time ($G$) &  $|V(G')|$  & $k$ & NetworkX Time ($G$) & Maximum Clique & NetworkX Time ($G'$) & $|K-\text{core}|$\\
     \hhline{|=|=|=|=|=|=|=|=|=|}
      \multirow{10}{*}{4000} & 0.1 & 81.62s & 2431 & 48 & 6397.12s & 5 & 2170.29s & 3928\\
    & 0.1 & 77.94s  & 2467 & 48 & 6245.25s & 5 & 2179.91s & 3944\\
    & 0.1 & 81.14s & 2624 & 48 & 6378.87s & 5 & 2222.77s & 3914\\
    & 0.1 & 81.25s & 2672 & 48 & 6309.69s & 5 & 2211.48s & 3929\\
    & 0.1 & 78.95s & 2711 & 49 & 6682.31s & 6 & 2321.84s & 3944\\
    & 0.1 & 79.66s & 2587 & 48 & 6219.77s & 5 & 2203.12s & 3939\\
    & 0.1 & 80.22s & 2701 & 48 & 6437.82s & 6 & 2290.86s & 3932\\
    & 0.1 & 81.45s & 2816 & 49 & 6598.12s & 5 & 2345.59s & 3937\\
    & 0.1 & 79.95s & 2722 & 48 & 6507.04s & 7 & 2298.02s & 3942\\
    & 0.1 & 80.55s & 2323 & 48 & 6384.13s & 5 & 2135.18s & 3928\\
    \hline
      \multirow{10}{*}{4000} & 0.3 & 89.85s & 2226 & 365 & Memory Error & 10 & 2254.02s & 3967 \\
    & 0.3 & 84.50s & 2600 & 365 & Memory Error & 10 & 2562.85s & 3960\\
    & 0.3 & 83.66s & 2390 & 363 & Memory Error & 11 & 2271.79s & 3952\\
    & 0.3 & 88.46s & 2678 & 365 & Memory Error & 10 & 2442.07s & 3975\\
    & 0.3 & 87.19s & 2474 & 365 & Memory Error & 10 & 2476.86s & 3976\\
    & 0.3 & 88.33s & 2141 & 364 & Memory Error & 11 & 2166.88s & 3977\\
    & 0.3 & 84.66s & 2649 & 365 & Memory Error & 10 & 2483.02s & 3970\\
    & 0.3 & 85.16s & 2563 & 365 & Memory Error & 11 & 2578.14s & 3978\\
    & 0.3 & 89.21s & 2741 & 365 & Memory Error & 10 & 2845.96s & 3983\\
    & 0.3 & 88.18s & 2554 & 364 & Memory Error & 10 & 2678.45s & 3971\\
    \hline
      \multirow{10}{*}{4000} & 0.4 & 81.24s & 2651 & 638 & Memory Error & NA &  Memory Error & 3975\\
    & 0.4 & 81.57s & 2787 & 638 & Memory Error & NA & Memory Error & 3979\\
    & 0.4 & 83.41s & 2893 & 638 & Memory Error & NA & Memory Error & 3967\\
    & 0.4 & 82.55s & 2532 & 638 & Memory Error & NA & Memory Error & 3972\\
    & 0.4 & 83.16s & 2746 & 638 & Memory Error & NA & Memory Error & 3975\\
    & 0.4 & 81.67s & 2778 & 638 & Memory Error & NA & Memory Error & 3972\\
    & 0.4 & 81.97s & 2766 & 638 & Memory Error & NA & Memory Error & 3975\\
    & 0.4 & 82.88s & 2699 & 638 & Memory Error & NA & Memory Error & 3989\\
    & 0.4 & 82.31s & 2863 & 639 & Memory Error & NA & Memory Error & 3958\\
    & 0.4 & 83.02s & 2453 & 638 & Memory Error & NA & Memory Error & 3977\\
    \hline
    \end{tabular}
    \caption{The table shows time taken using our graph reduction algorithm for $n=4000$ and $\rho(G)=0.1,0.3,0.4$ respectively. Time taken for NetworkX to find the maximum clique is also shown.}
    \label{tab:results for 4000}
\end{table}

\begin{table}[H]
    \centering
    \begin{tabular}{|c|c|p{0.1\linewidth}|c|c|p{0.1\linewidth}|p{0.1\linewidth}|p{0.1\linewidth}|}
    \hline
     $|V(G)|$  & $\rho(G)$ & Reduction Time ($G$) & $|V(G')|$  & $k$ & NetworkX Time ($G$) & Maximal Clique & NetworkX Time ($G'$) \\
     \hhline{|=|=|=|=|=|=|=|=|}
     \multirow{10}{*}{4000} & 0.4 & 81.24s & 2651 & 638 & 4223.10s & 11 & 1334.84s \\
    & 0.4 & 81.57s & 2787 & 638 & 4126.55s & 12 & 1316.78s \\
    & 0.4 & 83.41s & 2893 & 638 & 4245.37s & 11 & 1367.10s \\
    & 0.4 & 82.55s & 2532 & 638 & 4245.70s & 11 & 1311.05s \\
    & 0.4 & 83.16s & 2746 & 638 & 4188.53s & 11 & 1335.59s \\
    & 0.4 & 81.67s & 2778 & 638 & 4240.57s & 12 & 1369.93s \\
    & 0.4 & 81.97s & 2766 & 638 & 4178.86s & 11 & 1380.70s \\
    & 0.4 & 82.88s & 2699 & 638 & 4144.87s & 11 & 1321.22s \\
    & 0.4 & 82.31s & 2863 & 639 & 4315.92s & 12 & 1391.25s \\
    & 0.4 & 83.02s & 2453 & 638 & 4113.51s & 11 & 1205.51s \\
    \hline
    \end{tabular}
    \caption{The table shows the difference in time taken to find the maximal clique in the original graph $G$ compared to the reduced graph $G'$.}
    \label{tab:results for 4000 maximal}
\end{table}

\end{appendices}

\bibliographystyle{plainnat}
\bibliography{main}

\end{document}